\documentclass[12pt]{article}
\usepackage[margin=1in]{geometry}
\usepackage{amsmath,amssymb,amsthm}
\usepackage{booktabs}
\usepackage{caption}
\usepackage{graphicx}
\usepackage{hyperref}
\usepackage{placeins}
\usepackage{setspace}
\usepackage{xcolor}
\usepackage{url}
\newcommand{\1}{\mathbf{1}}

\newtheorem{proposition}{Proposition}

\theoremstyle{definition}

\makeatletter
\renewcommand\section{\@startsection {section}{1}{\z@}%
  {-3.5ex \@plus -1ex \@minus -.2ex}%
  {2.3ex \@plus .2ex}%
  {\normalfont\Large\bfseries\raggedright
   \hyphenpenalty=10000\exhyphenpenalty=10000}}
\renewcommand\subsection{\@startsection{subsection}{2}{\z@}%
  {-3.25ex\@plus -1ex \@minus -.2ex}%
  {1.5ex \@plus .2ex}%
  {\normalfont\large\bfseries\raggedright
   \hyphenpenalty=10000\exhyphenpenalty=10000}}
\renewcommand\subsubsection{\@startsection{subsubsection}{3}{\z@}%
  {-3.25ex\@plus -1ex \@minus -.2ex}%
  {1.5ex \@plus .2ex}%
  {\normalfont\normalsize\bfseries\raggedright
   \hyphenpenalty=10000\exhyphenpenalty=10000}}
\makeatother

\hypersetup{
  hidelinks,
  pdftitle={Routing Frictions and Executable Liquidity in Fragmented Markets},
  pdfauthor={Wen-Ting Wang},
  pdfsubject={Routing frictions and executable liquidity},
  pdfkeywords={market fragmentation, executable liquidity, venue activation,
  transaction costs, automated market makers, decentralized exchanges}
}

\begin{document}
\pagenumbering{roman}
\begin{titlepage}
\thispagestyle{empty}
{\raggedright\hyphenpenalty=10000\exhyphenpenalty=10000
\noindent{\Large\bfseries Routing Frictions and Executable Liquidity in\par}
\noindent{\Large\bfseries Fragmented Markets\par}}
\vspace{0.3in}
\noindent Wen-Ting Wang\par
\noindent Institute of Statistics, National Chung Hsing University, Taichung, Taiwan\par
\noindent ORCID: \href{https://orcid.org/0000-0003-3051-7302}{0000-0003-3051-7302}\par
\noindent Corresponding author. Email: \href{mailto:egpivo@gmail.com}{egpivo@gmail.com}
\vspace{0.5in}
\begin{abstract}
Public blockchains can make many trading venues simultaneously visible and
mechanically reachable, yet an order still has to pay to activate each
additional venue: technological connectivity need not translate into
economically integrated execution. Automated-market-maker (AMM) pools make
this gap directly measurable, because exact pre-trade venue states,
transaction-level routing costs, and realized venue use can be reconstructed
jointly from public blockchain records. Using 13,768 sampled
family-transactions from 29 Ethereum and Base token-pair families and 71
sibling pools, we document a sharp compression from gross to net execution
opportunity: the equal-chain lower bound on gross multi-venue gains is
34.85\%, while the corresponding upper bound after measured access costs is
6.89\%, as access costs eliminate 80.45\% of point-identified states
with positive gross gains. Realized routing shows a distinct second
contrast: actual multi-pool activation occurs in only 1.203\% of the opportunity
population, yet among 170 eligible realized integrated routes, observed
allocation captures 94.8\% of aggregate feasible gain. Holding recipient
orders and exact venue states fixed, replacing the Ethereum routing-cost
regime with Base's lower-cost regime materially expands executable
integration under both state populations. For digital-market design, the
results imply that blockchain scaling, cross-venue connectivity, and
displayed liquidity do not by themselves establish economically integrated
execution: integration is an order-specific property that depends on
transaction-level access costs.
\end{abstract}

\vspace{0.25in}

\vspace{0.12in}

\end{titlepage}

\pagenumbering{arabic}

\section{Introduction}

Public blockchains can make many trading venues simultaneously visible and
mechanically reachable. Anyone can query the exact reserves, price, and depth
of every pool trading a given asset pair, and a single transaction can in
principle route through several pools at once. Yet reaching an additional
venue is not free: it requires another contract call, additional state
access, and additional transaction data, each carrying a real, order-specific
cost. This is a digital-market question rather than a purely computational
one: does technological connectivity translate into executable market
integration, or can a market remain economically segmented at the point of
execution even when every venue is simultaneously visible and reachable?

This question is difficult to answer with conventional market data: the
exact pre-trade state of every alternative venue, the incremental cost of
reaching each one, and the venues an order actually used are typically held
by different parties, recorded in different systems, or not recorded at all
(Foucault \& Menkveld, 2008; Kwan et al., 2015; Battalio et al., 2016).
Programmable, transparent digital markets change this. Uniswap v3 fee tiers
create economically substitutable sibling pools for the same token pair,
whose execution curves and state transitions are public block by block
(Angeris et al., 2021; Adams et al., 2021; Lehar et al., 2024); transaction
receipts reveal exactly which pools an order used, and network gas and, on
Base, L1 data-fee schedules price the incremental cost of reaching an
additional venue. Public event order lets us reconstruct the state of every
sibling pool immediately before a transaction, so exact pre-trade venue
states, venue-specific execution curves, order size, incremental routing
cost, realized venue activation, and transaction ordering can all be
recovered jointly. This joint observability allows us to measure, transaction
by transaction, the wedge between what is technologically reachable and what
is economically worth reaching.

Technological integration is therefore not the same as economic integration,
and posted liquidity is not the same as order-level executable liquidity. A
second venue improves execution only when the gain from reaching it exceeds
the incremental cost of doing so, so prices can be closely linked and
liquidity can be visible across every venue while the market remains
segmented at the point of execution. We call the depth that survives this
comparison executable liquidity. Because it depends on the order, the
configuration of venue states when it arrives, and the incremental access
cost, executable liquidity is an order-specific object: two orders reaching
the same fragmented market can face different effective markets, one
profitably combining several venues while the other optimally behaves as if
only one venue existed. This is a classical market-fragmentation problem
(Pagano, 1989; Chowdhry \& Nanda, 1991; Madhavan, 1995; O'Hara \& Ye, 2011;
Malamud \& Rostek, 2017; Chen \& Duffie, 2021), recreated in a digital-market
environment where venue state, execution technology, and transaction-level
access cost are unusually observable.

The mechanism is direct. Fragmented sibling venues offer potentially
complementary depth, and combining them can raise gross execution value
above the best single venue. Activating an additional venue, however,
requires additional execution resources, so each order faces an endogenous
threshold: gross routing gain must exceed incremental access cost. Orders on
opposite sides of this threshold face different effective markets even when
the physical venue set is identical, and a lower access cost shifts the
threshold and expands the region in which fragmented liquidity is
executable. This separates two margins of market integration. The extensive
margin is whether another venue becomes worth activating. The intensive
margin is how the order is allocated once several venues are used. These
margins need not move together. Gross gains from fragmented depth may fail to
cover access cost, and measured net opportunity may still fail to appear in a
realized route. Conversely, a router that activates several venues may
allocate the order close to the feasible optimum. Distinguishing these
margins locates the measured friction without treating every unobserved route
as an allocation failure.

We reconstruct exact execution opportunities for 13,768 sampled
family-transactions from 29 Ethereum and Base token-pair families and 71
sibling pools. The data first show a sharp compression from gross to net
opportunity. The equal-chain lower bound on gross multi-venue opportunity is
34.85\%, while the corresponding net upper bound after measured access cost is
6.89\%. The conservative gap is 27.96 percentage points, with a 95\% interval
of 18.19 to 34.23 points. Among point-identified gross-positive states,
measured access cost removes 80.45\%, with an interval of 76.17\% to 88.62\%.

A further compression appears between modeled opportunity and realized venue
activation. Actual multi-pool integration occurs in 1.203\% of the full
opportunity population. Broad clean sibling integration occurs in 0.0446\%,
and strict verified integration in 0.0031\%. These rates do not identify the
content of the residual route-choice wedge; they show that measured
technological and economic opportunity rarely appears as observed multi-venue
use. Yet conditional on eligible activation, allocation is comparatively
efficient. Among 170 eligible actual integrated routes spanning 21 families,
observed allocation captures 94.8\% of aggregate feasible gain. The central
empirical result is therefore not simply that routing is costly. The largest
measured compression occurs at the extensive margin of venue activation,
rather than in allocation conditional on eligible use.

The paper makes three contributions. First, it introduces an order-level
empirical measure of
executable fragmentation in programmable markets, combining exact venue-state
reconstruction with transaction-specific access costs. Venue count,
consolidated depth, and cross-venue price linkage do not identify whether
marginal liquidity is worth accessing for a specific order; the contribution
is the empirical population mass around the routing-cost threshold, not the
threshold algebra itself.

Second, we separate the extensive margin, whether another venue becomes
economically worth activating, from the intensive margin, how efficiently the
order is allocated once multiple venues are activated. The main empirical
compression occurs at the extensive margin: actual multi-pool activation is
rare, while allocation conditional on eligible activation is comparatively
efficient. This is a more specific claim than saying that routers split
orders imperfectly; it locates the larger measured friction at whether a
venue is reached at all, rather than at how well it is used once reached.

Third, using Ethereum and Base cost regimes, we measure how lower
transaction-level access costs shift the executable-integration region while
holding recipient orders and exact venue states fixed. Replacing the Ethereum
cost regime with Base's lower regime raises lower-bound net integration by
20.88 percentage points in Ethereum states and 8.77 points in Base states,
each the weighted mass of orders between two measured cost thresholds. This
is a routing-cost comparative-static exercise, not a causal effect of chain
identity, a welfare estimate, or a claim that Base is a better venue: it is
evidence about how digital-market infrastructure costs map into executable
market integration. An exact-state transport accounting exercise confirms
that the result survives coherent transport of order size and venue design.

The analysis connects to three literatures, led by the one closest to the
setting. Research on AMM design, optimal routing, blockchain scaling, and
decentralized-exchange fragmentation provides the institutional and
computational foundation for the application and characterizes routing with
venue activation costs directly (Angeris et al., 2021, 2022; Adams et al.,
2021; Diamandis et al., 2023; Adams, 2024; Caparros et al., 2024; Lehar et
al., 2024; Lehar \& Parlour, 2025; Escudero et al., 2026). We measure how
often those costs change the economically executable set using
transaction-ordered historical states, and then connect that set to realized
routing. Classical work on market fragmentation and venue competition studies
liquidity externalities, multimarket trading, and market quality across
venues (Pagano, 1989; Chowdhry \& Nanda, 1991; Madhavan, 1995; Malamud \&
Rostek, 2017; O'Hara \& Ye, 2011; Degryse et al., 2015; Chen \& Duffie, 2021).
Research on order routing and trading rules shows that technological links do
not eliminate the effects of access, fees, rebates, and intermediary
incentives (Foucault \& Menkveld, 2008; Kwan et al., 2015; Battalio et al.,
2016; Cont \& Kukanov, 2017). Digital markets recreate this longstanding
fragmentation problem in a new institutional environment, one in which venue
state, execution technology, and transaction-level access costs are unusually
observable; the setting does not overturn the traditional literature but
makes one of its central objects directly measurable.

Section 2 defines executable liquidity and the two routing margins. Section 3
describes the sample, exact-state reconstruction, routing-cost measurement,
and dependence-adjusted inference. Sections 4 and 5 report executable
opportunity and realized routing. Section 6 presents the routing-cost
counterfactual and exact-state validation. Section 7 concludes. Online
Resource 1 reports detailed dependence designs, route-classification audits,
failed clean-route outcome tests, and weighting sensitivities.
Figure~\ref{fig:mechanism} summarizes the empirical sequence.

\begin{figure}[!htbp]
\centering
\includegraphics[width=\linewidth]{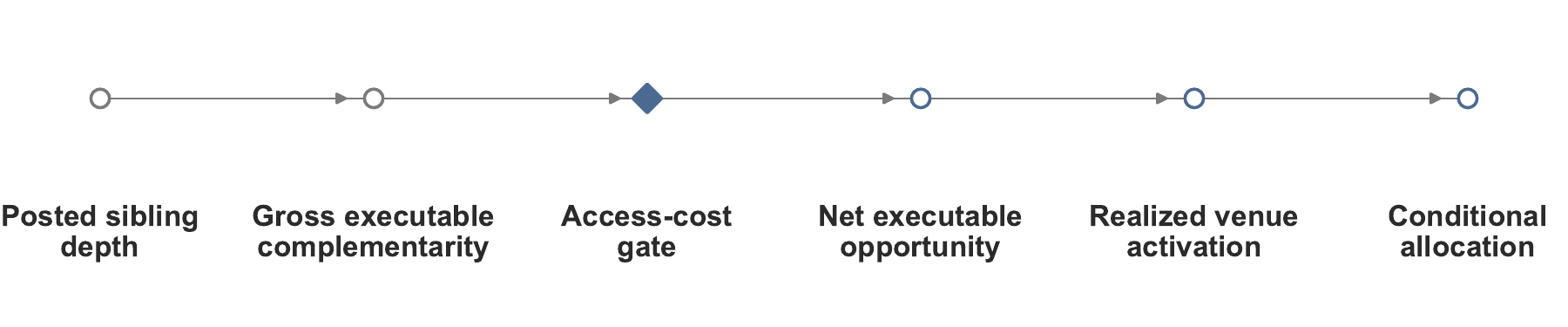}
\caption{From posted depth to realized allocation. Posted sibling liquidity
creates gross executable complementarity only when exact multi-venue execution
improves on the best single venue. The measured routing-cost gate determines
net executable opportunity. Realized venue activation and allocation
conditional on eligible use are separate empirical margins.}
\label{fig:mechanism}
\end{figure}

\FloatBarrier
\section{Executable Liquidity in Fragmented Markets}

\subsection{Sibling Venues and Routing Frictions}

An AMM pool is a venue with a state-dependent execution curve (Angeris et al.,
2021). Separate Uniswap v3 pools can exchange the same pair of assets while
charging different fees and applying fee-specific tick spacing (Adams et al.,
2021). We call these pools siblings. They are economically substitutable
venues, although their execution curves need not coincide.

A transaction may use one sibling or split input across several. Multi-pool
use requires additional contract calls, state access, and transaction data.
Ethereum charges execution gas. Base, an Ethereum layer-2, charges execution
gas plus an L1 data fee for posting transaction data back to Ethereum. We use
the term routing friction descriptively to denote this transaction-level cost
wedge between gross and net executable value; the term does not imply social
waste or a welfare gain from eliminating the cost.

The point of tracking this cost architecture is not that one chain is
cheaper than the other. It is that the technology underlying a digital market
sets the transaction-level price of reaching fragmented liquidity, and that
price is a design choice rather than a fixed feature of fragmentation itself.
A DEX router or liquidity aggregator that treats every economically
substitutable pool as reachable is implicitly assuming this price is low
enough to ignore; layer-2 scaling changes that assumption by moving the
access-cost schedule, not by removing fragmentation across pools. The
Ethereum/Base contrast is therefore informative about how digital-market
infrastructure choices map into effective market integration, independent of
any claim about which chain is preferable.

Public events identify price, active liquidity, initialized ticks, and
transaction order; receipts identify the realized pools and transfer paths.

\subsection{Gross and Net Executable Complementarity}

Let \(\mathcal J=\{1,\ldots,J\}\) denote economically substitutable venues.
An order submits input \(q>0\). Venue \(j\) returns \(V_j(q_j;s_j)\), in a
common pre-transaction numeraire, when it receives \(q_j\) in state \(s_j\).
Let
\begin{equation}
\mathcal Q(q)=\left\{\mathbf q\in\mathbb R_+^J:\sum_jq_j=q\right\},
\qquad S(\mathbf q)=\{j:q_j>0\}.
\end{equation}
For completeness, the executable value of an order under an activation-cost
schedule \(\kappa(S,z)\) is
\begin{equation}
\mathcal W(q,s,z)=\max_{\mathbf q\in\mathcal Q(q)}
\left\{\sum_jV_j(q_j;s_j)-\kappa(S(\mathbf q),z)\right\}.
\label{eq:fixed-charge-operator}
\end{equation}
The empirical application uses a binary incremental cost:
\begin{equation}
\kappa(S,z)=k(z)\1\{|S|\geq2\},\qquad k(z)\geq0,
\label{eq:binary-cost}
\end{equation}
where single-venue execution is the zero-cost benchmark.

Define the best single-venue value, unconstrained gross multi-venue value, and
gross gain as
\begin{equation}
B(q,s)=\max_jV_j(q;s_j),\qquad
I(q,s)=\max_{\mathbf q\in\mathcal Q(q)}\sum_jV_j(q_j;s_j),\qquad
G(q,s)=I(q,s)-B(q,s)\geq0.
\end{equation}

\paragraph{Routing-cost decision rule.}
Under \eqref{eq:binary-cost}, executable value is
\begin{equation}
\mathcal W(q,s,z)=\max\{B(q,s),I(q,s)-k(z)\}.
\label{eq:routing-rule}
\end{equation}
An integrated allocation is strictly preferred exactly when
\(G(q,s)>k(z)\). At equality, no integrated allocation is strictly
preferred. The rule defines the empirical distinction between gross
complementarity, \(G>0\), and net executable complementarity, \(G>k\).

\paragraph{Executable-integration region.}
For a cost regime \(c\) with schedule \(k^c(z)\), define
\begin{equation}
\mathcal I^c=\{(q,s,z):G(q,s)>k^c(z)\}.
\label{eq:executable-region}
\end{equation}
Executable integration is therefore order specific: two orders facing the
same fragmented venue set can lie on opposite sides of the routing-cost
threshold. Under regime \(c\), the corresponding net-integration rate is the
design-weighted empirical mass of order-state-cost realizations inside
\(\mathcal I^c\). When gross gain is bounded, the data identify lower and upper
masses for this region.

For transaction \(i\), the design-weighted incidence rates are
\begin{equation}
\theta^G=\frac{\sum_iw_i\1\{G_i>0\}}{\sum_iw_i},\qquad
\theta^N(z)=\frac{\sum_iw_i\1\{G_i>k_i(z)\}}{\sum_iw_i}.
\end{equation}
When execution identifies only
\(\underline G_i\leq G_i\leq\overline G_i\), replace \(G_i\) in each
indicator by its lower or upper endpoint to obtain
\(\underline\theta^G,\overline\theta^G,\underline\theta^N\), and
\(\overline\theta^N\).

\begin{proposition}[Incidence bounds under interval-valued execution]
\label{prop:sharp-bounds}
If the only transaction-level restriction is
\(\underline G_i\leq G_i\leq\overline G_i\), with \(w_i\geq0\), these four
endpoint rates are sharp bounds on gross and net integration incidence. In
particular,
\begin{equation}
\theta^G-\theta^N(z)\geq
\underline\theta^G-\overline\theta^N(z).
\end{equation}
\end{proposition}

The bounds preserve deterministic reconstruction uncertainty rather than
dropping unresolved executions. Sampling uncertainty is handled separately.

\subsection{Venue Activation and Conditional Allocation}

The extensive margin is whether a transaction activates more than one
economically substitutable sibling venue. The intensive margin is how
efficiently it allocates input conditional on eligible observed multi-venue
use. Broad and strict clean routes are nested, higher-specificity subsets of
the population route census. The efficiency sample is a separate conditional
sample spanning route objectives; it is not nested within the clean-route
subsets. Figure~\ref{fig:universes} makes these denominators explicit.

\FloatBarrier
\begin{figure}[!htbp]
\centering
\includegraphics[width=\linewidth]{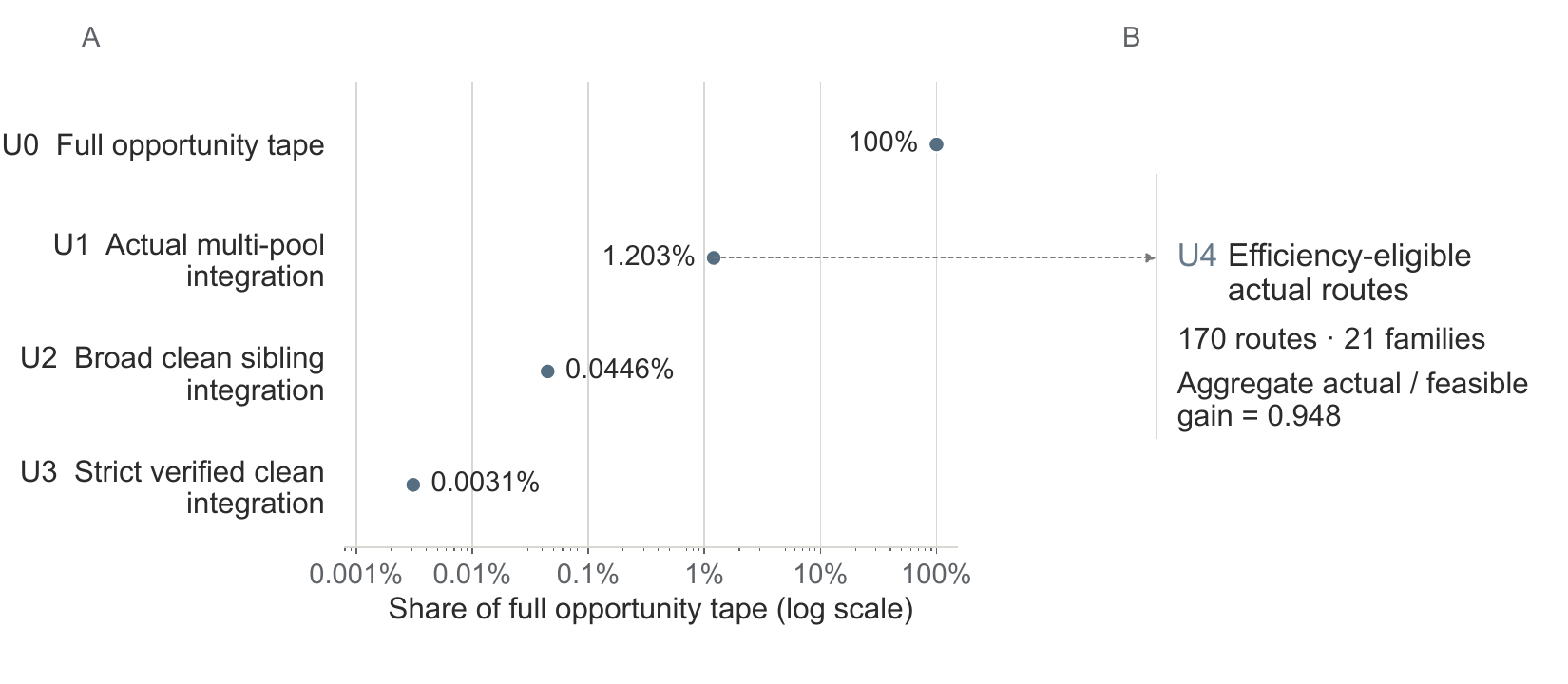}
\captionsetup{labelsep=period}
\caption{Empirical universes. Panel A reports population route rates using the
full 13,768-transaction opportunity tape. Panel B reports allocation
efficiency for a separate conditional sample of 170 eligible actual integrated
routes spanning 21 families. The efficiency sample branches from U1 across
route objectives and is not nested within U2 or U3.}
\label{fig:universes}
\end{figure}

\FloatBarrier
\section{Data and Exact Execution Measurement}

\subsection{Sample and Family Construction}

The analysis covers 29 prespecified token-pair families on Ethereum and Base
from January 2024 through December 2025. A family contains economically linked
sibling pools for one token pair; a family-transaction is a sampled
transaction evaluated against that family's executable siblings. The sample
contains 7,176 Ethereum and 6,592 Base transactions. Deterministic hash
sampling gives every included transaction a known inclusion probability.

Eligibility is fixed using token semantics, sustained activity, and
sibling-pool availability. Gross gain, net gain, realized routes,
convergence, and markouts do not enter selection. Population estimates retain
inverse-inclusion weights. Across 71 pools, 12,817 transactions are point
identified, 917 are economic no-overlap zeros, and 34 retain deterministic
bounds. Economic zeros and bounded rows remain in the reported denominators.
Table~\ref{tab:summary} reports sample construction and the main estimates.

\begin{table}[!htbp]
\centering
\caption{Sample and Main Estimates}
\label{tab:summary}
\small
\begin{tabular}{@{}p{0.59\linewidth}r@{}}
\toprule
Measure & Estimate \\
\midrule
\multicolumn{2}{@{}l}{\textit{Panel A. Sample}} \\
Token-pair families & 29 (15 Ethereum; 14 Base) \\
Sibling pools & 71 \\
Sampled family-transactions & 13,768 (7,176 Ethereum; 6,592 Base) \\
Point identified & 12,817 \\
No-overlap economic zeros & 917 \\
Bounded & 34 \\
\addlinespace
\multicolumn{2}{@{}l}{\textit{Panel B. Opportunity and activation}} \\
Gross-complementary share, lower bound & 34.85\% \\
Net-complementary share, upper bound & 6.89\% \\
Conservative gross--net gap & 27.96 pp \\
Gross-positive states removed by access cost & 80.45\% \\
Actual multi-pool use & 1.203\% \\
Broad clean sibling use & 0.0446\% \\
Strict verified use & 0.0031\% \\
\addlinespace
\multicolumn{2}{@{}l}{\textit{Panel C. Conditional allocation}} \\
Eligible actual integrated routes & 170 \\
Families & 21 \\
Design-weighted route efficiency & 0.912 \\
Aggregate actual/feasible gain & 0.948 \\
Under-routing rate & 20.96\% \\
\bottomrule
\end{tabular}

\parbox{0.96\linewidth}{\footnotesize \textit{Note:} Panel A reports counts on the full prespecified opportunity tape. Panel B uses that tape as the denominator. Gross lower and net upper are conservative bounds; removal is measured among point-identified gross-positive states. Panel C is conditional on 170 eligible actual integrated routes across 21 families. Route efficiency is design weighted, aggregate actual/feasible is the ratio of aggregate gains. Under-routing denotes an observed allocation whose feasible gross-gain shortfall exceeds the frozen tolerance, $F_i-A_i>\tau_i$.}

\end{table}

\subsection{Transaction-Ordered Venue States}

Within each block, events are ordered by transaction index and log index. The
pre-state for transaction \(i\) contains all earlier state changes in the
block and no later changes. For each executable sibling, the reconstruction
retains price, liquidity, fee, tick, and initialized-tick data.

The execution engine simulates exact execution across initialized ticks and
searches over feasible allocations. It reports best single-pool output,
multi-pool output, pool allocations, and tick crossings. Observed routes are
replayed against receipt transfers. Deterministic measurement failures enter
lower and upper bounds rather than being silently deleted.

\subsection{Routing-Cost Measurement}

Incremental execution gas is estimated on receipt-validated one- and
multi-sibling routes. The frozen specification yields a transaction- and
date-specific full-router cost schedule using observed gas prices and
fee-regime parameters rather than a constant cost level. Base additionally
includes incremental calldata and applicable L1 cost. ETH/USD and gas prices
are determined before the focal execution outcome is evaluated. The schedule
measures the incremental on-chain cost of activating multiple siblings; it
does not measure every private or institutional routing friction.

\subsection{Dependence-Adjusted Inference}

The primary procedure resamples families within chain and circular contiguous
eight-week calendar blocks within family. It redraws gas coefficients,
ETH/USD, gas prices, and Base calldata/L1 parameters while preserving
transaction weights and deterministic execution bounds. Headline opportunity
estimates use 5,000 valid draws per chain. Route-rate, eligible-efficiency,
and realized-outcome intervals retain family, calendar-date, exact-block, and
family-block dependence, with wild-family inference or descriptive-only
reporting for small-family cells. Alternative block designs are sensitivity
checks. Transaction-iid inference is diagnostic only.

\section{Routing Frictions and Executable Integration}

\subsection{Gross versus Net Opportunity}

The equal-chain lower bound on gross-complementary states is 34.85\%. After
subtracting the frozen routing-cost schedule, the upper bound on net
complementarity is 6.89\%. The conservative gross-lower-minus-net-upper gap is
27.96 percentage points, with a primary interval of 18.19 to 34.23 points.
Thus posted sibling depth often creates a gross execution improvement, but
far fewer orders can cover the measured cost of reaching that depth.

\subsection{Threshold-Band Removal}

Let \(\mathcal P\) be the point-identified set. The weighted share of
gross-positive states removed by measured access cost is
\begin{equation}
\rho=
\frac{\sum_{i\in\mathcal P}w_i\1\{0<G_i\leq k_i\}}
{\sum_{i\in\mathcal P}w_i\1\{G_i>0\}}
=1-\frac{\sum_{i\in\mathcal P}w_i\1\{G_i>k_i\}}
{\sum_{i\in\mathcal P}w_i\1\{G_i>0\}}.
\end{equation}
Measured access cost removes 80.45\% of these states, with a primary interval
of 76.17\% to 88.62\%. The 34 bounded rows remain in the separate gross and
net incidence bounds rather than being assigned a removal status.

\subsection{Cross-Family Evidence}

Figure~\ref{fig:family-gross-net} reports gross and net incidence bounds by
selected family.

\begin{figure}[!htbp]
\centering
\includegraphics[width=\linewidth]{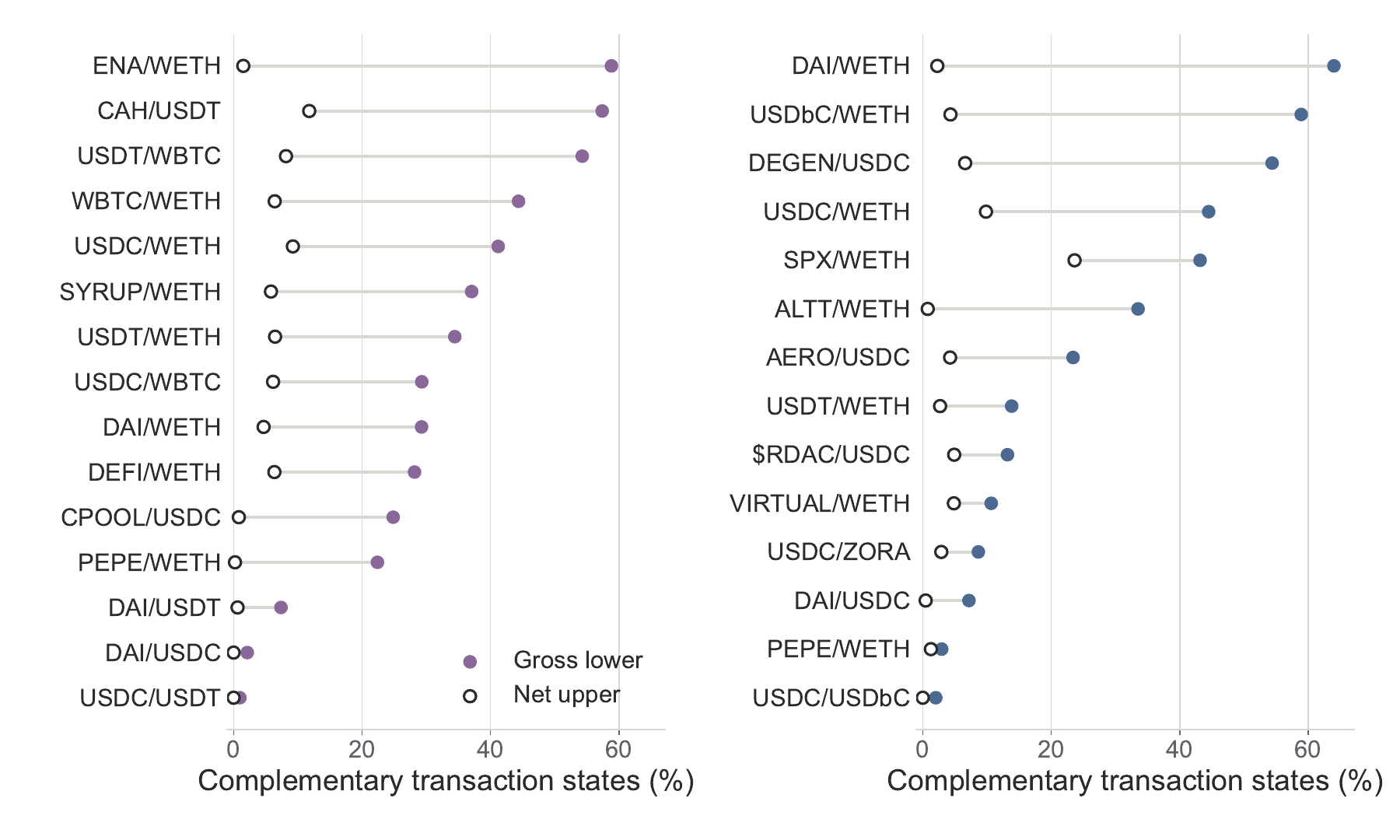}
\caption{Gross versus net integration across families. The left panel reports
Ethereum and the right panel Base. Filled markers show the design-weighted
lower bound on gross-complementary states; hollow markers show the upper bound
after subtracting the central routing-cost schedule. Economic no-overlap
states remain zero and bounded execution rows remain in the reported bounds.}
\label{fig:family-gross-net}
\end{figure}

USDC--WETH is economically important but does not carry the result. Removing
that family on both chains yields an 83.54\% gross-state removal share and a
75.23\% to 90.96\% interval in the dedicated eight-week exclusion bootstrap.
The exclusion result is not described as passing block-length sensitivities
that were not run.

These estimates locate an economically important friction before realized
route choice. Positive modeled net gain is an order-level execution
opportunity; it is not realized trader surplus and does not identify welfare.

\section{Realized Routing}

\subsection{Route Census}

The route taxonomy distinguishes clean sibling splits from aggregator
execution, arbitrage or backruns, protocol operations, and complex
multi-objective transactions. A broad clean split uses only economically
linked sibling pools. The strict measure additionally requires verified
router evidence. Labels are fixed before realized-outcome analysis and use
decoded calldata, receipt logs, call traces, token-flow graphs, and transfer
reconciliation rather than Swap-log count alone.

\subsection{Population Venue-Activation Rates}

Population rates use the full 13,768-transaction opportunity tape. Actual
multi-pool integration occurs in 1.203\% of that population. Broad clean
sibling integration is a nested route class and occurs in 0.0446\%; strict
verified integration occurs in 0.0031\%. The broad clean two-way family/date
interval is 0 to 0.0982\%. These are different route-taxonomy thresholds
against the same population denominator.

The gap between net opportunity and realized activation can contain search
cost, command constraints, urgency, rebates, bundling, slippage protection,
private routing value, trader heterogeneity, and measurement error. The
current design does not identify those components or the full latent
route-choice wedge.

\subsection{Conditional Allocation Efficiency}

For a transfer-reconciled observed allocation
\(\mathbf q_i^{\mathrm{obs}}\), define observed and feasible gross gains
relative to the best single venue as
\begin{equation}
A_i=\sum_jV_j(q_{ij}^{\mathrm{obs}};s_{ij})-B_i,
\qquad
F_i=\max_{\mathbf q\in\mathcal Q(q_i)}\sum_jV_j(q_j;s_{ij})-B_i=G_i.
\end{equation}
Both use the same input and pre-state. The same binary access cost \(k_i\) is
subtracted from observed and feasible routes. Let
\begin{equation}
\tau_i=\max\{0.01\ \mathrm{USD},\;10^{-5}\times
\text{input value}_i\}.
\end{equation}
The tolerance prevents rounding-scale or economically negligible deviations
from being classified as under-routing.
For a point-identified, transfer-reconciled actual integrated route satisfying
\(F_i-k_i\geq\tau_i\), route efficiency is
\begin{equation}
e_i=\frac{A_i-k_i}{F_i-k_i}.
\end{equation}
The design-weighted estimate averages \(e_i\). The aggregate ratio divides
\(\sum_iw_i(A_i-k_i)\) by \(\sum_iw_i(F_i-k_i)\). Under common input, state,
and access cost, \(A_i\leq F_i\), so \(e_i\leq1\) when the denominator is
positive. Values above one are retained only as reconstruction diagnostics.

For the eligible route set \(\mathcal E\), under-routing is
\begin{equation}
U_i=\1\{F_i-A_i>\tau_i\},\qquad
U=\frac{\sum_{i\in\mathcal E}w_iU_i}
{\sum_{i\in\mathcal E}w_i}.
\end{equation}
Thus under-routing denotes a feasible gross-gain shortfall exceeding the
frozen tolerance; it never compares the dimensionless ratio \(e_i\) with the
dollar threshold \(\tau_i\).

Among 170 eligible routes in 21 families, design-weighted route efficiency is
0.912, with a family/date interval of 0.887 to 0.938. Aggregate actual gain is
0.948 of aggregate feasible gain, with an interval of 0.900 to 0.995. The
under-routing rate is 20.96\%. The two statistics answer different questions:
under-routing counts routes with a shortfall above \(\tau_i\), whereas the
aggregate ratio weights each route by its feasible net gain. These conditional
allocation results use a different universe from the clean-route population
rates.

Read together with the population activation rates in Section 5.2, this
result places the larger measured wedge before allocation: routes that clear
the extensive margin are, on average, allocated close to the feasible
optimum, so the dominant compression in the data lies in whether a venue is
activated at all rather than in how well it is used once activated. For
digital routing infrastructure, this suggests that improving post-activation
split optimization has limited scope to close the gap documented in Section
5.2, whose components remain unidentified in this design.

\subsection{Post-Trade Outcomes as Interpretation Boundaries}

Pooled all-route family-relative markouts are negative at one and five blocks
and remain a secondary descriptive fact. No universal clean-route convergence
or markout sign survives the relevant dependence adjustment; Base clean-route
convergence has only three contributing families. The Supplementary Appendix
reports these tests. We therefore do not use post-trade outcomes to strengthen the
venue-activation result.

\section{Routing-Cost Counterfactual Across Market Environments}

This section treats the routing-cost regime as a comparative-static object.
Holding the recipient order and the exact pre-trade venue state fixed, we ask
how much of the executable-integration region changes when only the cost of
activating additional venues changes. This is a stronger empirical exercise
than a cross-sectional Ethereum--Base correlation, because order and
venue-state composition cannot mechanically drive the result.

\subsection{Common Support}

Raw Ethereum--Base differences combine routing cost, order flow, venue states,
and fee design. The counterfactual instead changes the measured routing-cost
regime while keeping the recipient order and exact venue state fixed. An
exact-stratum design has limited common support: it retains eight Ethereum and
74 Base observations from two families per chain, with maximum normalized
weights above 14\%. We therefore define the comparison population using
outcome-blind propensity-overlap weights based on pair type, economic
direction, quarter, order size, active virtual liquidity, sibling dominance,
executable sibling count, pre-state volatility, and missingness. The resulting
population retains all 29 families. Effective sample sizes are 909 on Ethereum
and 309 on Base, maximum normalized weights are 1.14\% and 1.79\%, and the
largest absolute standardized covariate difference is 0.0075. For comparison,
the Supplementary Appendix describes the exact-stratum design, and Table S1
reports weighting sensitivities.

The exact order-size and venue-design transport requires recipient orders to
be executable on complete coherent donor states. Its narrower validation
population contains 3,478 Ethereum and 3,497 Base target transactions and
retains every family. The largest absolute standardized difference is 0.0214.

\subsection{Cost-Regime Transport}

For target chain \(t\in\{E,B\}\), let \(\mathcal O_t\) be the overlap
population and let \(c\in\{E,B\}\) index the routing-cost regime. Under the
prespecified design and overlap weights \(w_{it}\), lower-bound net integration is
\begin{equation}
\underline\Theta_{t,c}=
\frac{\sum_{i\in\mathcal O_t}w_{it}
\1\{\underline G_{it}>k_{it}^{\,c}\}}
{\sum_{i\in\mathcal O_t}w_{it}}.
\label{eq:cost-matrix}
\end{equation}
Target order, exact state, fee tiers, gross-gain bounds, and weights do not
change across columns. Missing donor dates are matched only within the same
calendar week.

\begin{proposition}[Ordered-cost threshold mass]
\label{prop:threshold-mass}
For \(k_i^L\leq k_i^H\) in a fixed weighted population, holding orders and
venue states fixed,
\begin{equation}
\underline\Theta_t(k^L)-\underline\Theta_t(k^H)
=\frac{\sum_{i\in\mathcal O_t}w_{it}
\1\{k_{it}^L<\underline G_{it}\leq k_{it}^H\}}
{\sum_{i\in\mathcal O_t}w_{it}}.
\end{equation}
\end{proposition}

On matched support, the Base cost is no greater than the Ethereum cost for
each transaction. Restricted to that support, every realization classified as
executable under the higher-cost regime is also classified as executable
under the lower-cost regime. A lower routing-cost schedule therefore weakly
expands the executable-integration region on the comparison support.
Proposition~\ref{prop:threshold-mass} measures the design-weighted lower-bound
empirical mass reclassified as net executable when the higher schedule is
replaced by the lower schedule while orders and venue states are held fixed.
The result requires no parametric distribution for gains.
Figure~\ref{fig:cost-counterfactual} reports the contrast.

\begin{figure}[!htbp]
\centering
\makebox[\linewidth][c]{%
\includegraphics[width=1.12\linewidth]{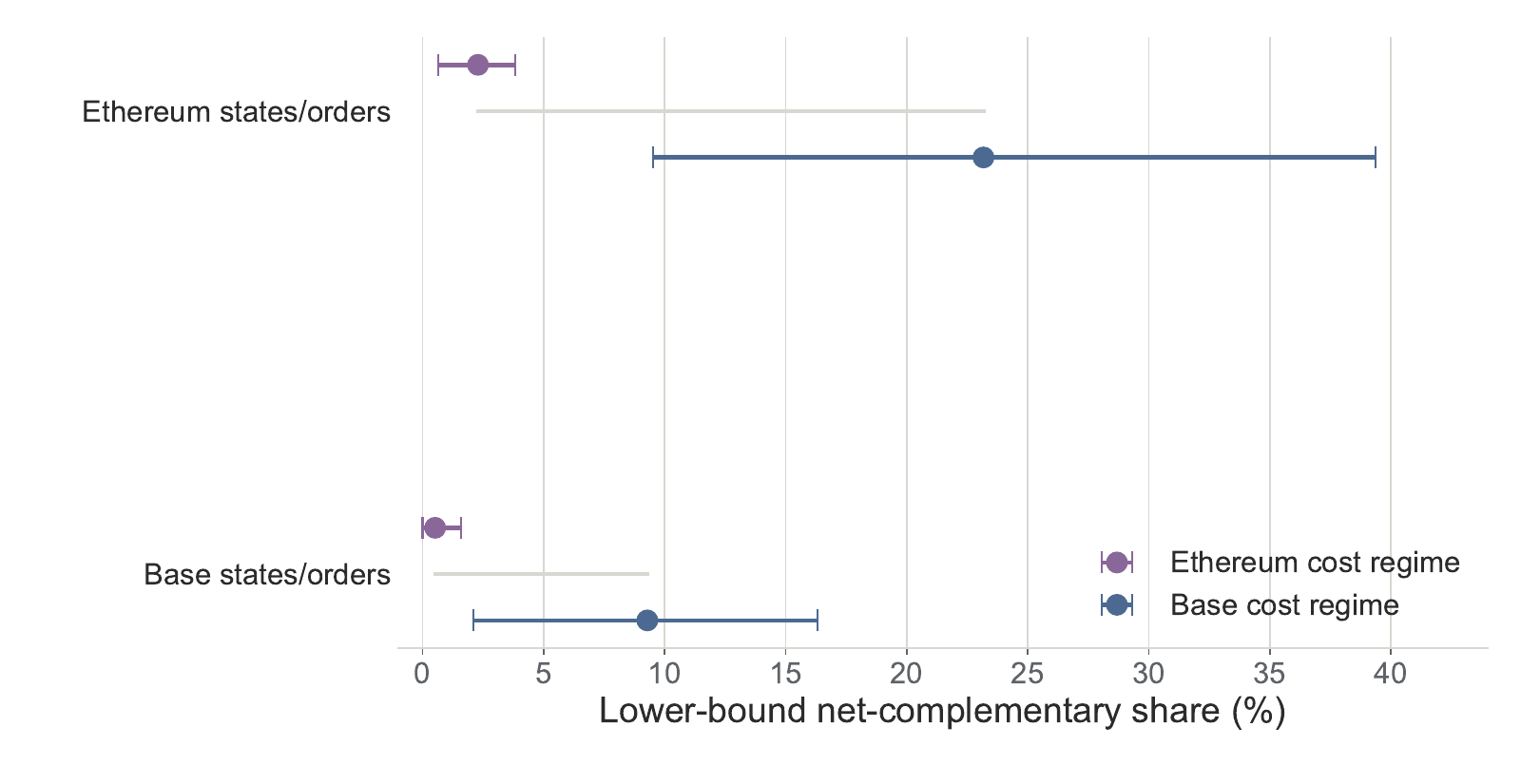}}
\caption{Routing-cost counterfactual on balanced support. Within each state
and order population, the two markers change only the routing-cost regime.
Lines are 95\% intervals from 5,000 family and circular contiguous eight-week
common-shock draws. The comparison is not a chain treatment effect.}
\label{fig:cost-counterfactual}
\end{figure}

Under Ethereum states and orders, replacing the Ethereum cost regime with
Base's raises lower-bound net integration from 2.29\% to 23.17\%, a difference
of 20.88 percentage points with a 95\% interval of 7.81 to 38.04 points. Under
Base states and orders, the same replacement raises the share from 0.52\% to
9.28\%, a difference of 8.77 points with an interval of 1.84 to 16.05 points.
Both lower endpoints remain positive under all six alternative block designs.

\subsection{Exact-State Validation}

Recovered histories cover 71 pools and 184,699,451 ordered events. A
checkpoint-plus-delta cache reconstructs 33,551 marker-pool states for all
13,768 transaction markers. Every frozen transaction reproduces with zero
gross-classification mismatch. Exact counterfactual execution retains bounds
for 61 point-unidentified executions rather than dropping them or coding them
as zero. The maximum transport-decomposition adding-up error is
\(2.84\times10^{-14}\).

These checks rule out proportional output scaling, local tick approximation,
summary-state interpolation, incompatible fee--liquidity hybrids, and
optimizer approximation as explanations for the cost result.

\subsection{K/Q/D Accounting}

The validation transports routing cost \(K\), order-size distribution \(Q\),
and joint venue design \(D\). The order-size map preserves USD size, economic
direction, pair type, and common-support rank before exact re-execution. The
venue-design map moves the complete donor execution object: prices, active
liquidity, initialized-tick books, liquidity placement, fee tiers, tick
spacing, sibling count, and dominance. Each transported coalition is
reoptimized. The Appendix defines the full Shapley accounting operator and
its exact adding-up property.

Figure~\ref{fig:kqd} reports the contributions. It is a validation and
accounting exercise rather than the paper's primary cross-market estimand.

\begin{figure}[!htbp]
\centering
\makebox[\linewidth][c]{%
\includegraphics[width=1.12\linewidth]{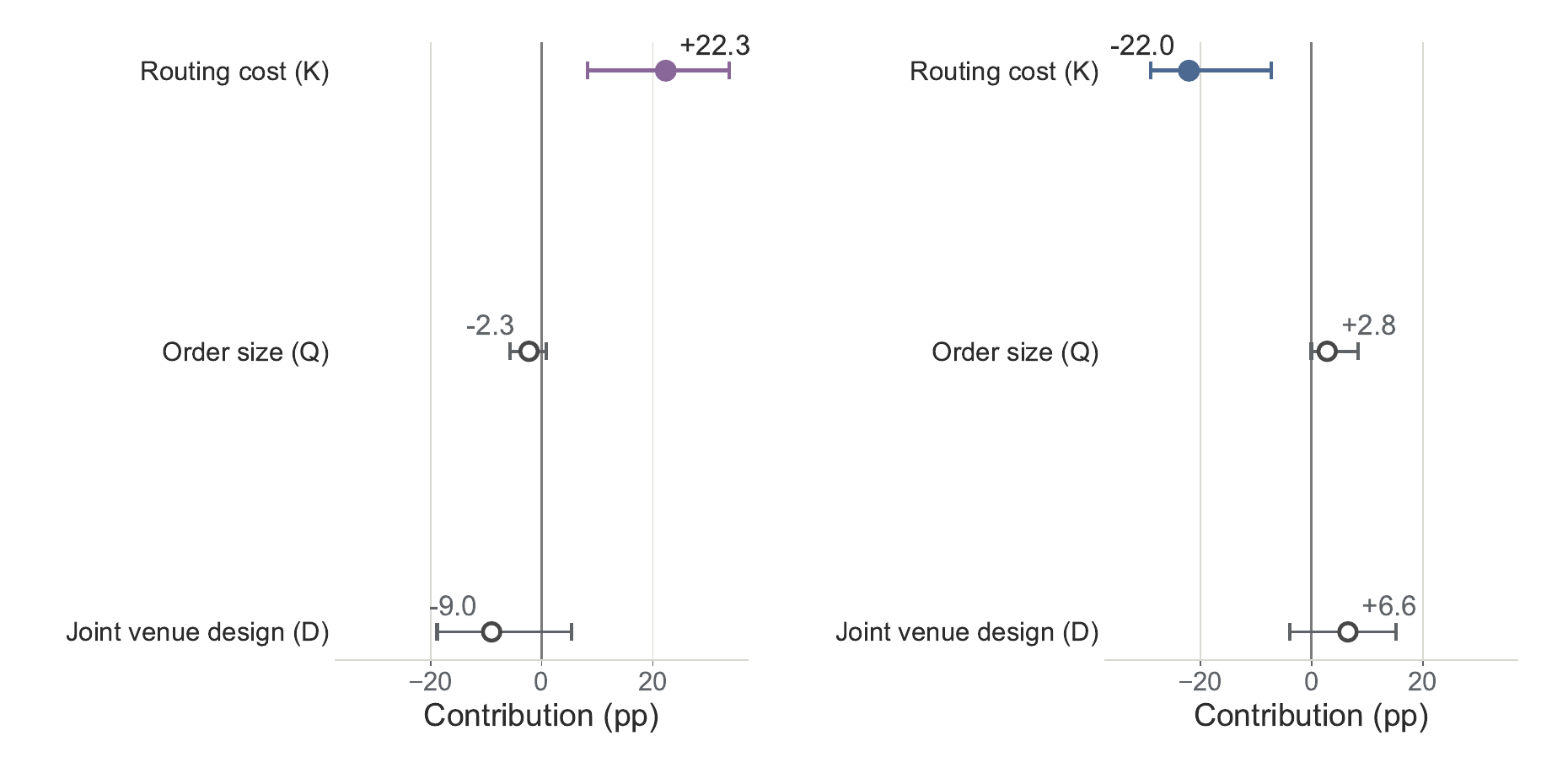}}
\caption{Exact K/Q/D accounting on coherent transport support. The left panel
targets Ethereum and the right panel targets Base. Filled markers show the
dependence-robust K terms; hollow markers show Q and D descriptive accounting
terms. Intervals use 5,000 family and circular contiguous eight-week
common-shock draws. The zero line marks no contribution; the figure does not
estimate a causal effect of chain identity.}
\label{fig:kqd}
\end{figure}

For the Ethereum target, \(K\) contributes \(+22.30\) percentage points, with
an interval of \([8.26,33.70]\); \(Q\) contributes \(-2.29\) points and \(D\)
\(-9.01\) points. For the Base target, \(K\) contributes \(-22.03\) points,
with an interval of \([-28.96,-7.23]\); \(Q\) contributes \(+2.83\) points
and \(D\) \(+6.56\) points. Only \(K\) supports a dependence-robust standalone
claim. The Q and D intervals and sensitivity patterns do not support stable
standalone attribution. The opposite signs reflect transport direction:
Base's lower cost expands the Ethereum target, whereas Ethereum's higher cost
contracts the Base target.

Fee tier, tick spacing, and liquidity placement remain a joint venue-design
object. A separate fee-only hybrid would require an economically identified
pool assignment that the cross-chain data do not supply. A related
protocol-fee design can identify LP-supply responses when trader-facing fees
are unchanged (Wang, 2026), but only 24 of 6,975 observations here pass the
mechanical fee-tier screen. We therefore do not introduce a fourth transported
component.

\section{Conclusion}

Fragmented markets can be technologically connected while remaining
economically segmented for individual orders. In the sibling-pool setting,
venues are simultaneously visible and mechanically reachable, yet measured
access costs remove about four-fifths of point-identified gross-positive
states. Actual multi-pool activation occurs in only 1.203\% of the
opportunity population, while in a separate conditional sample of eligible
realized routes, observed allocation captures 94.8\% of aggregate feasible
gain. The evidence places the largest measured compression before
allocation, at venue activation.

The framework interprets these results through an order-specific
executable-integration region. A routing-cost regime partitions order-state
realizations according to whether the gross gain from reaching additional
displayed liquidity exceeds its incremental access cost. Aggregate net
integration summarizes the empirical mass inside that region rather than
describing a single effective market faced uniformly by every order. Holding
recipient orders and exact venue states fixed, applying Base's lower cost
regime expands lower-bound net integration under both state populations; this
is a comparative statement about the executable set rather than a causal
claim about chain identity.

Three implications follow for digital-market design. Blockchain transparency
and composability can make venues technologically connected without making
them economically integrated for every order: a fragmented market can look
integrated in aggregate while presenting different effective markets to
different orders. Scaling that lowers transaction-level access cost can
change the set of liquidity that is economically executable, not merely
reduce an average transaction fee; the routing-cost comparative statics in
Section 6 quantify this as a shift in the executable-integration region
rather than a uniform saving. And router design should be evaluated
separately at venue activation and at conditional allocation. In our setting,
allocation is comparatively efficient once eligible multi-venue routes are
realized, while activation is rare, so improving split optimization after
venues are activated has limited scope to address the larger measured
compression, whose remaining components are not identified here.

Venue availability, consolidated depth, and price linkage describe
properties of a fragmented market, but they do not determine which liquidity
remains economically reachable for a particular order. Digital-market
integration is, in this sense, an order-specific economic object rather than
merely a property of the venue network.

\clearpage
\section*{Statements and Declarations}
\addcontentsline{toc}{section}{Statements and Declarations}

\subsection*{Funding}
Wen-Ting Wang received funding from the National Science and Technology
Council, Taiwan (Grant ID NSTC 113-2118-M005-005-MY2).

\subsection*{Competing Interests}
The author declares no known competing financial interests or personal
relationships that could have appeared to influence the work reported in
this paper.

\subsection*{Data Availability}
The raw data analysed in this study are publicly available Ethereum and Base
blockchain records, accessible via public node providers and block explorers
(e.g., Etherscan, Basescan). The processed transaction tape and reconstructed
pool-state histories generated during the current study are available from
the corresponding author upon reasonable request.

\subsection*{Generative AI and AI-Assisted Technologies in Manuscript
Preparation}
During the preparation of this work, the author used OpenAI Codex to assist
with manuscript organization, language editing, LaTeX preparation, code-level
formatting of data visualizations, and submission-compliance checks. Empirical
results and specifications were produced by author-controlled code, and all
figures were rendered deterministically from frozen numerical inputs using
author-reviewed Matplotlib scripts; no generative image model was used. After
using this service, the author reviewed and edited the content as needed and
takes full responsibility for the content of the article.

\clearpage
\section*{References}
\addcontentsline{toc}{section}{References}
\begingroup
\setlength{\parindent}{0pt}
\setlength{\parskip}{0.7\baselineskip}
\hangindent=1.5em
\hangafter=1
Adams, A. (2024). Layer 2 be or layer not 2 be: Scaling on Uniswap v3.
arXiv preprint arXiv:2403.09494.
\url{https://doi.org/10.48550/arXiv.2403.09494}.

\hangindent=1.5em
\hangafter=1
Adams, H., Zinsmeister, N., Salem, M., Keefer, R., \& Robinson, D. (2021).
Uniswap v3 core [Technical report].

\hangindent=1.5em
\hangafter=1
Angeris, G., Chiang, R., Chitra, T., Kao, H.-T., \& Noyes, C. (2021). An
analysis of Uniswap markets. \textit{Cryptoeconomic Systems}.

\hangindent=1.5em
\hangafter=1
Angeris, G., Evans, A., Chitra, T., \& Boyd, S. (2022). Optimal routing for
constant function market makers. In \textit{Proceedings of the 23rd ACM
Conference on Economics and Computation} (pp. 115--128).

\hangindent=1.5em
\hangafter=1
Battalio, R., Corwin, S. A., \& Jennings, R. (2016). Can brokers have it all?
On the relation between make-take fees and limit order execution quality.
\textit{Journal of Finance, 71}, 2193--2238.

\hangindent=1.5em
\hangafter=1
Caparros, B., Chaudhary, A., \& Klein, O. (2024). Blockchain scaling and
liquidity concentration on decentralized exchanges. arXiv preprint
arXiv:2306.17742. \url{https://doi.org/10.48550/arXiv.2306.17742}.

\hangindent=1.5em
\hangafter=1
Chen, D., \& Duffie, D. (2021). Market fragmentation. \textit{American
Economic Review, 111}, 2247--2274.

\hangindent=1.5em
\hangafter=1
Chowdhry, B., \& Nanda, V. (1991). Multimarket trading and market liquidity.
\textit{Review of Financial Studies, 4}, 483--511.

\hangindent=1.5em
\hangafter=1
Cont, R., \& Kukanov, A. (2017). Optimal order placement in limit order
markets. \textit{Quantitative Finance, 17}, 21--39.

\hangindent=1.5em
\hangafter=1
Degryse, H., de Jong, F., \& van Kervel, V. (2015). The impact of dark trading
and visible fragmentation on market quality. \textit{Review of Finance, 19},
1587--1622.

\hangindent=1.5em
\hangafter=1
Diamandis, T., Resnick, M., Chitra, T., \& Angeris, G. (2023). An efficient
algorithm for optimal routing through constant function market makers. In F.
Baldimtsi \& C. Cachin (Eds.), \textit{Financial cryptography and data
security} (Lecture Notes in Computer Science, Vol. 13951, pp. 128--145).
Springer. \url{https://doi.org/10.1007/978-3-031-47751-5_8}.

\hangindent=1.5em
\hangafter=1
Escudero, C., Lara, F., \& Sama, M. (2026). Optimal routing across constant
function market makers with gas fees. arXiv preprint arXiv:2603.02844.
\url{https://doi.org/10.48550/arXiv.2603.02844}.

\hangindent=1.5em
\hangafter=1
Foucault, T., \& Menkveld, A. J. (2008). Competition for order flow and smart
order routing systems. \textit{Journal of Finance, 63}, 119--158.

\hangindent=1.5em
\hangafter=1
Kwan, A., Masulis, R. W., \& McInish, T. H. (2015). Trading rules, competition
for order flow and market fragmentation. \textit{Journal of Financial
Economics, 115}, 330--348.

\hangindent=1.5em
\hangafter=1
Lehar, A., \& Parlour, C. A. (2025). Decentralized exchange: The Uniswap
automated market maker. \textit{Journal of Finance, 80}, 321--374.

\hangindent=1.5em
\hangafter=1
Lehar, A., Parlour, C. A., \& Zoican, M. (2024). Fragmentation and optimal
liquidity supply on decentralized exchanges. arXiv preprint arXiv:2307.13772.
\url{https://doi.org/10.48550/arXiv.2307.13772}.

\hangindent=1.5em
\hangafter=1
Madhavan, A. (1995). Consolidation, fragmentation, and the disclosure of
trading information. \textit{Review of Financial Studies, 8}, 579--603.

\hangindent=1.5em
\hangafter=1
Malamud, S., \& Rostek, M. (2017). Decentralized exchange.
\textit{American Economic Review, 107}, 3320--3362.

\hangindent=1.5em
\hangafter=1
O'Hara, M., \& Ye, M. (2011). Is market fragmentation harming market quality?
\textit{Journal of Financial Economics, 100}, 459--474.

\hangindent=1.5em
\hangafter=1
Pagano, M. (1989). Trading volume and asset liquidity. \textit{Quarterly
Journal of Economics, 104}, 255--274.

\hangindent=1.5em
\hangafter=1
Wang, W.-T. (2026). Causal effects of protocol-fee changes on liquidity
provision in automated market makers. arXiv preprint arXiv:2607.08525.
\url{https://doi.org/10.48550/arXiv.2607.08525}.
\endgroup

\appendix
\clearpage
\setcounter{figure}{0}
\renewcommand{\thefigure}{A.\arabic{figure}}
\setcounter{table}{0}
\renewcommand{\thetable}{A.\arabic{table}}
\setcounter{equation}{0}
\renewcommand{\theequation}{A.\arabic{equation}}
\setcounter{section}{0}
\section{Formal Results and Transport Accounting}

\subsection{Routing-Cost Decision Rule}

Every single-venue allocation has value at most \(B(q,s)\). Every allocation
supported on at least two venues has net value at most \(I(q,s)-k(z)\). If
\(I(q,s)>B(q,s)\), every maximizer of \(I(q,s)\) uses at least two venues, so
\(I(q,s)-k(z)\) is attainable. If \(I(q,s)=B(q,s)\), then
\(I(q,s)-k(z)\leq B(q,s)\). Hence
\[
\mathcal W(q,s,z)=\max\{B(q,s),I(q,s)-k(z)\}.
\]
Since \(I(q,s)=B(q,s)+G(q,s)\), an integrated alternative is strictly larger
than the best single venue exactly when \(G(q,s)>k(z)\). At equality, no
integrated allocation is strictly preferred.

\subsection{Proof of the Incidence Bounds}

\begin{proof}[Proof of Proposition~\ref{prop:sharp-bounds}]
For any scalar threshold \(a_i\) and every admissible
\(G_i\in[\underline G_i,\overline G_i]\),
\[
\1\{\underline G_i>a_i\}
\leq \1\{G_i>a_i\}
\leq \1\{\overline G_i>a_i\}.
\]
Multiplication by \(w_i\geq0\), summation, and normalization give the stated
bounds. Taking \(a_i=0\) gives gross incidence; taking \(a_i=k_i(z)\) gives
net incidence. Each endpoint is attainable from the interval information
alone by selecting the lower or upper admissible gain transaction by
transaction. The bounds are therefore sharp absent cross-transaction
restrictions. Finally, \(\theta^G\geq\underline\theta^G\) and
\(\theta^N(z)\leq\overline\theta^N(z)\), which imply the conservative-gap
inequality.
\end{proof}

\subsection{Proof of the Ordered-Cost Identity}

\begin{proof}[Proof of Proposition~\ref{prop:threshold-mass}]
For each observation, \(k_i^L\leq k_i^H\) implies
\[
\1\{\underline G_i>k_i^L\}-\1\{\underline G_i>k_i^H\}
=\1\{k_i^L<\underline G_i\leq k_i^H\}.
\]
The result follows after multiplying by the fixed weight, summing over the
target population, and dividing by total weight.
\end{proof}

\subsection{Exact K/Q/D Accounting}

For target population \(t\), donor population \(d\), and component set
\(\mathcal M=\{K,Q,D\}\), let \(T_S^{t\leftarrow d}\) transport the coherent
subset \(S\subseteq\mathcal M\) and re-solve the exact allocation problem.
The coalition value is
\begin{equation}
v_{t\leftarrow d}(S)=
\underline\theta^N\!\left(T_S^{t\leftarrow d}(\mathcal P_t)\right).
\end{equation}

\begin{proposition}[Exact coherent-transport accounting]
\label{prop:shapley}
For \(m\in\mathcal M\), define
\begin{equation}
\phi_{m,t\leftarrow d}
=\sum_{S\subseteq\mathcal M\setminus\{m\}}
\frac{|S|!\,(|\mathcal M|-|S|-1)!}{|\mathcal M|!}
\left[v_{t\leftarrow d}(S\cup\{m\})-v_{t\leftarrow d}(S)\right].
\end{equation}
Then
\begin{equation}
\sum_{m\in\mathcal M}\phi_{m,t\leftarrow d}
=v_{t\leftarrow d}(\mathcal M)-v_{t\leftarrow d}(\varnothing).
\end{equation}
\end{proposition}

The operator transports complete executable objects and reoptimizes after
each coalition. The contributions are directional accounting quantities
conditional on the common-support map, not structural or causal effects of
the transported components.

\begin{proof}[Proof of Proposition~\ref{prop:shapley}]
For any permutation of \(\mathcal M\), summing components' marginal
contributions telescopes from \(v_{t\leftarrow d}(\varnothing)\) to
\(v_{t\leftarrow d}(\mathcal M)\). Averaging over permutations gives the
stated contributions and adding-up identity. Reoptimization defines every
coalition value, so the result does not require additive execution technology.
\end{proof}

\end{document}